\documentclass[11pt]{article}
\usepackage{fullpage}
\usepackage{amsfonts}
\usepackage{amssymb}
\usepackage{amstext}
\usepackage{amsmath}
\usepackage{amsthm}
\usepackage{thmtools}
\usepackage{graphicx}
\usepackage{graphics}
\usepackage{colordvi}
\usepackage{color}
\usepackage{paralist}
\usepackage{standalone}
\usepackage{bm}
\usepackage{todonotes}
\usepackage{mathtools}
\usepackage{thm-restate}
\usepackage{tikz}
\usepackage{comment}
\usepackage[ruled,vlined]{algorithm2e}

\usepackage{soul}

\usepackage{hyperref}
\hypersetup{
  pdftitle = {Hardness and Approximation for Coloring Digraphs},
  pdfauthor = {},
  colorlinks = true,
  linkcolor = black!30!red,
  citecolor = black!30!green
}
\usepackage{cleveref}

\newtheorem{theorem}{Theorem}[section]
\newtheorem{corollary}[theorem]{Corollary}
\newtheorem{observation}[theorem]{Observation}
\newtheorem{lemma}[theorem]{Lemma}
\newtheorem{claim}{Claim}
\newtheorem{definition}[theorem]{Definition}

\newtheorem{conjecture}{Conjecture}

\newtheorem{question}{Question}

\newcommand{\eps}{\varepsilon}

\title{Hardness and Approximation for Coloring Digraphs}
\author{}

\date{May 18, 2026}

\author{Parinya Chalermsook\thanks{University of Sheffield, United Kingdom} \and Harmender Gahlawat\thanks{Department of Mathematics, Indian Institute of Technology Dehli, India} \and Felix Klingelhoefer\thanks{G-SCOP, Grenoble-INP, France} \and Alantha Newman\thanks{CNRS and LIP, ENS Lyon, France} \and Chaoliang Tang\thanks{Fudan University, China}}

\newcommand{\ndecomp}{path decomposition\xspace}
\newcommand{\pt}{T}
\newcommand{\NFP}{N^+_{t}}
\newcommand{\NFM}{N^-_{t}}

\newcommand{\fl}{f_{l}}

\newenvironment{cproof}
{\begin{proof}
 [Proof.]
 \vspace{-1.5\parsep}
}
{ \end{proof}}

\newcommand{\eset}{{\mathcal E}}

\usepackage{complexity}

\newcommand{\chiv}{\vec\chi}
\newcommand{\alphav}{\vec\alpha}
\newcommand{\BG}{G_{\prec}}

\begin{document}

\maketitle

\begin{abstract}
The \textit{dichromatic number} $\chiv(D)$ of a digraph is the minimum
number $k$ such that $V(D)$ can be partitioned into $k$ subsets, each
inducing an acyclic digraph. The \textit{acyclic number} $\alphav(D)$
is the cardinality of a largest induced acyclic subdigraph of $D$.

We study these problems from an approximation point of view. We begin
with establishing that even when restricted to tournaments,
approximating $\chiv$ and $\alphav$ remain as challenging as their
undirected counterparts on general graphs. Specifically, we establish
that for every $\epsilon >0$, it is hard to approximate both $\alphav$
and $\chiv$ up to a factor of $n^{1-\epsilon}$ even when restricted to
tournaments.

We next consider approximate coloring of digraphs in special cases. We
begin with establishing that we can color $\ell$-dicolorable digraphs
using at most $\ell \cdot n^{1-\frac{1}{\ell}}$ colors in time
$O(n^{2\ell})$; in particular, we can color $2$-dicolorable digraphs
with $2\sqrt{n}$ colors in polynomial time.  We then focus on bounding
the dichromatic number of dense digraphs as a function of the
independence number $\alpha$ of the underlying graph.  We consider two
special cases in this regard: digraphs with $\chiv(D)\leq 2$ and
digraphs that do not contain any directed triangle.  For these cases,
we present algorithms which generalize and improve existing tools and
results.

\end{abstract}

\section{Introduction}

Let $D=(V,A)$ be a digraph with vertex set $V$ and arc set $A$.  The
minimum number of induced acyclic sets into which the vertices of $D$
can be partitioned is known as its {\em dichromatic
  number}~\cite{neumann1982dichromatic}.  We use $\chiv(D)$ to denote
the dichromatic number of $D$. If $\chiv(D) \leq k$, we say that $D$
is {\em $k$-dicolorable} or simply {\em $k$-colorable}.  The
dichromatic number of a digraph is related to the chromatic number of
a graph.  Consider a simple undirected graph $G$, and replace each
edge with a directed 2-cycle (also known as a {\em digon}) to obtain a
digraph $D$ on the same vertex set.  Now fix a $k$-coloring that is
proper in $G$ (i.e., the endpoints of each edge have different
colors).  Then each color class forms an independent set in $G$ and an
induced acyclic set in $D$.  Hence, the coloring problem in undirected
graphs and the dicoloring problem in directed graphs in which each
pair of vertices is either connected by a digon or not connected at
all are equivalent.  This directly implies that for a digraph $D$, the
problem of computing $\chiv(D)$ and the problem of computing a maximum
induced acyclic set, denoted by $\alphav(D)$, are \NP-hard.

However, if we consider an arbitrary directed graph or even an {\em
  oriented} graph (in which digons are not allowed), then the direct
equivalence between dicoloring a digraph and coloring a graph
described above no longer applies.  A well-studied class of oriented
graphs is {\em tournaments}.  While the underlying undirected graph is
complete and thus has chromatic number equal to the number of
vertices, the dichromatic number can be as low as one if the
orientation of the arcs results in an acyclic or transitive
tournament.

In the graph theory literature, the problem of proving upper bounds on
the dichromatic number of digraphs has been extensively
studied~\cite{mohar2010eigenvalues,berger2013tournaments,aboulker2024heroes,aboulker2024overrightarrow},
motivated by its connection to the Erd\H{o}s-Hajnal
conjecture~\cite{erdos1989ramsey,Chu14}.  This conjecture states that
for any graph $H$, there is an absolute constant $\gamma_H$ such that
an $H$-free graph $G$ has a clique or independent set of size
$n^{\gamma_H}$, where $n$ is the number of vertices in $G$.  For
example, by a famous algorithm of
Wigderson~\cite{wigderson1983improving}, a triangle-free graph has an
independent set of size at least $\sqrt{n}$.

This conjecture has an analogous formulation for tournaments, which
states that for any tournament $H$, there is an absolute constant
$\tau_H$ such that an $H$-free tournament $T$ has an induced acyclic
set of size at least $n^{\tau_H}$, where $n$ is the number of vertices
in $T$~\cite{alon2001ramsey}.  The conjecture has been proved for many
classes of $H$-free tournaments.  For example, {\em heroes} are the
set of all tournaments whose exclusion results in constant dichromatic
number and linear sized induced acyclic sets, and
\cite{berger2013tournaments} gave a complete description of all
heroes.  Another example is tournaments that excludes any particular
six vertex tournament~\cite{chudnovsky2024pure,nguyen2023induced}.

Much of the work in the graph theory arena is based on proving upper
bounds by establishing existential results.  There has also been some
investigation into this topic from the algorithmic or complexity
perspective.  As noted above, computing the value $\chiv(D)$ is
\NP-hard.  Moreover, it is \NP-hard to decide if an oriented graph is
2-colorable~\cite{bokal2004circular}, to decide if a tournament is
2-colorable~\cite{chen2007min} or to decide if a tournament is
$k$-colorable~\cite{fox2019removal}. In fact, it is \NP-hard to decide
if a digraph $D$ is $2$-colorable even if it contains two vertices
$u,v$ such that $D\setminus{\{u,v\}}$ is acyclic (i.e., $\alphav(D)
\geq V(D)-2$)~\cite{harutyunyan2024digraph}.  Recently, the complexity
of approximating the dichromatic number in tournaments was
studied~\cite{klingelhoefer2024coloring}.  For example, it was shown
that one can color a 2-colorable tournament with ten colors and that
it is \NP-hard to color a 2-colorable tournament with three colors.
It was also shown that coloring 3-colorable tournaments and
3-colorable graphs are very close in terms of their complexity.  They
also showed that it is \NP-hard to approximate the dichromatic number
of a tournament within a factor of $n^{1/2-\delta}$ for any $0 <
\delta < 1/2$, which is a weaker state of affairs than what is known
for graph coloring.

\subsection{Our Results}

We address the dicoloring problem in general digraphs from several
different angles.  In Section \ref{S:hardness}, we address the
hardness of approximation of dicoloring an oriented digraph.  We show
strong hardness by proving that even in the ``easiest'' setting of
tournaments, approximate dicoloring is as hard as approximate graph
coloring. Our main result (\Cref{thm: main-hardness}) shows that, for
every $\epsilon >0$, it is hard to distinguish between the tournaments
$T$ with $\alphav(T) \leq |V(T)|^{\epsilon}$ and those with $\chiv(T)
\leq |V(T)|^{\epsilon}$. This implies $n^{1-\epsilon}$ hardness of
approximation for both $\alphav$ (\Cref{C:hardness-large-Digraph}) and
$\chiv$ (Corollary \ref{C:hardness-coloring}).

This motivates the problem of coloring digraphs in special cases,
addressed in Section \ref{sec:algorithms}, where we address
algorithmic approaches.  A well-studied classic problem in graph
coloring is to color a 3-colorable graph with few
colors~\cite{wigderson1983improving,karger1998approximate,kawarabayashi2024better}.
A natural, analogous question in the setting of digraphs, which does
not appear to have been previously addressed, is: How many colors do
we need to color a 2-dicolorable digraph in polynomial time?  We show
that we can color a 2-dicolorable $n$-vertex digraph with at most
$2\sqrt{n}$ colors (Theorem \ref{T:2-approximate}).  We generalize
this, showing how to color an $\ell$-dicolorable digraph with at most
$\ell \cdot n^{1 - \frac{1}{\ell}}$ colors in time $O(n^{2\ell})$
(Theorem \ref{T:l-approximate}).

Next, we study digraphs with bounded independence number, which are
sometimes called ``dense'' digraphs.  The independence number of a
digraph is the size of the maximum independent set in the underlying
undirected graph.  For example, tournaments are a well-studied class
of digraphs and have independence number one.  Digraphs with
independence number at most $\alpha$ are dense since they contain
(roughly) at least $n^2/(2\alpha)$ arcs.  (To see this, apply Turan's
theorem to the complement of the underlying undirected graph, which is
$K_{\alpha}$-free.)

We consider classes with additional (natural) restrictions such as
those promised to be 2-dicolorable or $C_3$-free, where $C_3$ is a
directed triangle.  Notice that a $C_3$-free tournament is acyclic.
For these restricted classes, we can prove upper bounds on the
dichromatic number via polynomial-time algorithms, based on a
non-trivial extension of the path decomposition for tournaments used by \cite{klingelhoefer2024coloring}.  
For example, we show that 2-dicolorable digraphs can be
efficiently colored with $\frac{10}{3}(4^\alpha -1)$ colors (\Cref{thm:2-col_dig}), which in particular generalizes the result of \cite{klingelhoefer2024coloring} for tournaments.\footnote{We remark that these algorithms do not need to know a maximum independent set of the underlying graph nor its size; with a light re-engineering, on any input digraph $D$ and any fixed integer $\alpha$, the algorithm can be made to either return a coloring with the claimed bound or return an independent set of size $\alpha+1$.}\footnote{Some of this material has appeared in \cite{FelixThesis}.}
Then we show that $C_3$-free digraphs can be colored with at most $\frac{(\alpha + 8)!}{9!}$ colors (Theorem \ref{thm:c3-color}), which improves on the upper bound of $35^{\alpha-1}\alpha!$ due to \cite{harutyunyan2019coloring}.  We also pose a conjecture (Conjecture \ref{AlphaConj}) that would imply a positive resolution to a recently posed open problem~\cite{openBarbados}, and prove it for a non-trivial case.

\setcounter{section}{1}

\section{Hardness Results}\label{S:hardness}

Recall that $\vec{\alpha}(D)$ and $\chiv(D)$ denote the size of maximum acyclic set and dichromatic number of $D$, respectively. 
We prove the following result which implies the hardness of approximating both $\vec{\chi}$ and $\vec{\alpha}$. 
\begin{restatable}{theorem}{mainHardness}\label{thm: main-hardness}
For $\epsilon \in (0,1)$, there is no polynomial-time algorithm that, on input tournament $T$, can distinguish between the following two cases: 
\begin{itemize}
    \item (Completeness:) $\vec{\chi}(T) \leq n^{\epsilon}$ (which implies that $\vec{\alpha}(T) \geq n^{1-\epsilon}$)  

    \item (Soundness:) $\vec{\alpha}(T) \leq n^{\epsilon}$ (which implies that $\vec{\chi}(T) \geq n^{1-\epsilon}$).  
\end{itemize}
unless ${\sf NP} = {\sf RP}$. 
\end{restatable}

This implies the following corollaries:

\begin{restatable}{corollary}{hardnessLargeDigraph}\label{C:hardness-large-Digraph}
For $\epsilon \in (0,1)$, it is hard to approximate the maximum acyclic set of a tournament to within a factor of $n^{1-\epsilon}$ unless ${\sf NP} = {\sf RP}$.     
\end{restatable}

\begin{restatable}{corollary}{hardnessColoring}\label{C:hardness-coloring}
For $\epsilon \in (0,1)$, it is hard to approximate the dichromatic number of a tournament to within a factor of $n^{1-\epsilon}$ unless ${\sf NP} = {\sf RP}$.     
\end{restatable}

\paragraph{High-level overview:} Roughly speaking, the best known hardness reduction~\cite{feder2019complexity} (that gives a factor of $n^{1/2-\delta}$ hardness) combines a random bipartite graph with a hard instance of graph coloring, but their reduction blows up the size of the instance quadratically while maintaining the hardness factor. In particular, given a hard instance $G$ for chromatic number, they create a tournament $T$ of size roughly $|V(T)|\approx |V(G)|^2$, while preserving the hardness factor, i.e., $\chiv(T) \approx \chi(G)$. Since the hardness gap for the chromatic number is $|V(G)|^{1-\epsilon}$, we have that the hardness factor for the dichromatic number is $|V(G)|^{1-\epsilon} = |V(T)|^{1/2 -O(\epsilon)}$. Note that the loose factor is caused by the growth of the instance size while the hardness gap does not grow proportionally.  

To prove a tight hardness factor, we need a reduction that establishes a tighter relation. Our reduction takes hard instance $G$ for chromatic number and creates tournament $T'$ such that $|V(T')| \approx |V(G)|^{k}$ for $k = 1/\epsilon$; the instance size is, in fact, much larger than the reduction of~\cite{feder2019complexity}. However, in our case, the hardness factor grows proportionally with the size, i.e., the hardness gap is $\approx |V(G)|^{(1-\epsilon)k}$ which is simply a tight hardness of $|V(T')|^{1-\epsilon}$.    

Our reduction follows the high-level scheme of~\cite{chalermsook2013graph,chalermsook2014pre}, which uses the (lexicographic) graph product inequalities to tightly relate the optimal of the problem at hand to the stability number or the chromatic number. The new idea in this work is to introduce a natural extension of the standard lexicographic product that can be used in the directed setting (all previous works rely on the standard notion of graph products).  Our extension is randomized, allowing us to naturally combine nice properties of a random bipartite graph~\cite{feder2019complexity} with hard instances of graph coloring~\cite{feige1998zero}, in a way that the hardness gap grows proportionally to the size of the reduction. 
Our new graph product might be of independent interest.

\subsection{Main Reduction Lemma}

Key to proving the hardness result is the following lemma that simultaneously relates (i) the chromatic number to the dichromatic number of a tournament and (ii) the stability number to the size of a largest acyclic set.

\begin{lemma}\label{lem:reduction}
For all $k \in {\mathbb N}$, 
there exists a polynomial time randomized reduction that, on (sufficiently large) input undirected graph $\widehat{G}$, produces a tournament $\pt$ such that $|V(\pt)| = |V(\widehat{G})|^k$  and 
\begin{itemize}
    \item $\vec{\chi}(\pt) \leq \chi(\widehat{G})^k |V(\widehat{G})|$ with probability $1$.   

    \item $\vec{\alpha}(\pt) = O(\alpha(\widehat{G})^k |V(\widehat{G})|^2)$ with probability at least $1-1/k$ 
\end{itemize}
\end{lemma}

Before proving this lemma, we first show how this implies the hardness result in Theorem~\ref{thm: main-hardness}. We start with the following hardness result of Feige and Kilian~\cite{feige1998zero} (and derandomized by~\cite{zuckerman2006linear}).  

\begin{theorem}
\label{feige}
Let $\gamma \in (0,1)$. Given an undirected $N$-vertex graph $\widehat{G}$, it is NP-hard to distinguish between the following two cases: 
\begin{itemize}
    \item (Completeness:) $\chi(\widehat{G}) \leq N^{\gamma}$. 
    \item (Soundness:) 
    $\alpha(\widehat{G}) \leq N^{\gamma}$.  
\end{itemize}
\end{theorem}

To prove Theorem~\ref{thm: main-hardness}, consider $\epsilon < \epsilon_0$ for $\epsilon_0$ that will be chosen later. 
Our reduction takes an input graph $\widehat{G}$ from Theorem~\ref{feige} with parameter $\gamma = \epsilon/5$ and invokes Lemma~\ref{lem:reduction} with parameter $k = \lceil 1/\gamma \rceil$ to get  tournament $\pt$ with $n = |V(\pt)| = N^k$ vertices. 
In the completeness case, we know that $\vec{\chi}(\pt)  \leq \chi(\widehat{G})^k |V(\widehat{G})| \leq N^{k\gamma+1} =n^{\gamma+1/k} \leq n^{\epsilon}$. This happens with probability $1$.  
In the soundness case, we have that $\vec{\alpha}(\pt) \leq O(\alpha(\widehat{G})^k N^2) \leq O(N^{\gamma k +2})= O(N^{k(\gamma+2/k)}) = O(n^{\gamma+2/k})$; notice that $\gamma+2/k \leq \epsilon/5+\frac{2}{\lceil 5/\epsilon \rceil} \leq \epsilon$ for sufficiently small $\epsilon$ (we choose threshold $\epsilon_0$ such that it holds). This happens with probability at least $1- 1/k$. 
This reduction implies that any algorithm that distinguishes between $\vec{\chi}(\pt) \leq n^{\epsilon}$ and $\vec{\alpha}(\pt) \leq n^{\epsilon}$ would give us an RP algorithm that distinguishes between $\chi(\widehat{G}) \leq N^{\gamma}$ and $\alpha(\widehat{G}) \leq N^{\gamma}$, hence implying that ${\sf NP} = {\sf RP}$. This concludes the proof of Theorem~\ref{thm: main-hardness}.

\subsection{Proof of Lemma~\ref{lem:reduction}}

\subsubsection{Graph Products}

We first recall the graph product operation in undirected graphs. 
Let $G$ and $H$ be undirected graphs. Denote by $G \cdot H$ the \textit{lexicographic product} of $G$ and $H$, defined as follows: 
\[V(G \cdot H) = V(G) \times V(H) \] 
We write vertices of $G \cdot H$ in the coordinate form $(u,a)$ where $u \in V(G)$ and $a \in V(H)$. Edges of the product graph are defined as: 
\[E(G\cdot H) = \{\{(u,a),(v,b)\}: \{u,v\} \in E(G) \vee (u=v)\wedge \{a,b\} \in E(H) \} \]
In other words, there is an edge between a pair $(u,a)$ and $(v,b)$ iff there is an edge in the first coordinate ($\{u,v\} \in E(G)$), or if $u$ and $v$ are the same, $\{a,b\} \in E(H)$. 
It would be useful to think of the product as, first, replacing each vertex $u$ by a copy $H_u$ of graph $H$, and then connect $H_u$ and $H_v$ via a complete bipartite graph for those $\{u,v\} \in E(G)$.  
We will use the following well-known fact about chromatic and stability numbers of product graphs. 

\begin{theorem}
For any graphs $G$ and $H$, we have 
\begin{itemize}
    \item $\chi(G \cdot H) \leq \chi(G)  \chi(H)$, and 

    \item $\alpha(G \cdot H) = \alpha(G) \alpha(H)$. 
\end{itemize}
\end{theorem}

\subsubsection{Directed Randomized Lexicographic Product}

Let $H$ be a digraph and $G$ be an undirected graph (that we call a \textbf{skeleton}). Moreover, we work with a fixed ordering $\sigma: [|V(G)|] \rightarrow V(G)$ of vertices in $G$, i.e., $(\sigma(1), \sigma(2),\ldots, \sigma(|V(G)|))$ is a permutation of $V(G)$. 
The random process $D^{\sigma}_G[H]$ generates a product graph $G'$
via the following steps: (i) First, start with the vertex set
$V(G)$. Replace each vertex $v \in V(G)$ with a copy $H_v$ of digraph
$H$. For each $v$, we denote by ${\sf cloud}(v)$ the vertex set in the
copy $H_v$, that is, ${\sf cloud}(v) = \{(v,a): a \in V(H)\}$, (ii)
Next, for each $\{u,v\} \in E(G)$, and for all pairs of 
vertices $a,b \in V(H)$ add either an arc $(u,a)(v,b)$ or $(v,b)(u,a)$
with equal probability (i.e., we pick a random orientation for this
pair), and (iii) For each edge $\{u,v\} \not \in E(G)$ such that $u$
appears before $v$ in permutation $\sigma$ and for all pairs of 
vertices $a,b \in V(H)$, add arc $(u,a)(v,b)$. Another way to describe this random process is to first perform an undirected lexicographic product and then orient the $G$-edges between the clouds at random, while the edges inside the clouds are oriented according to $H$. See~\Cref{fig:lex1} and~\Cref{fig:lex2} for illustration of Steps (i) and (ii).   

\begin{figure}
    \centering   \includegraphics[width=0.8\linewidth]{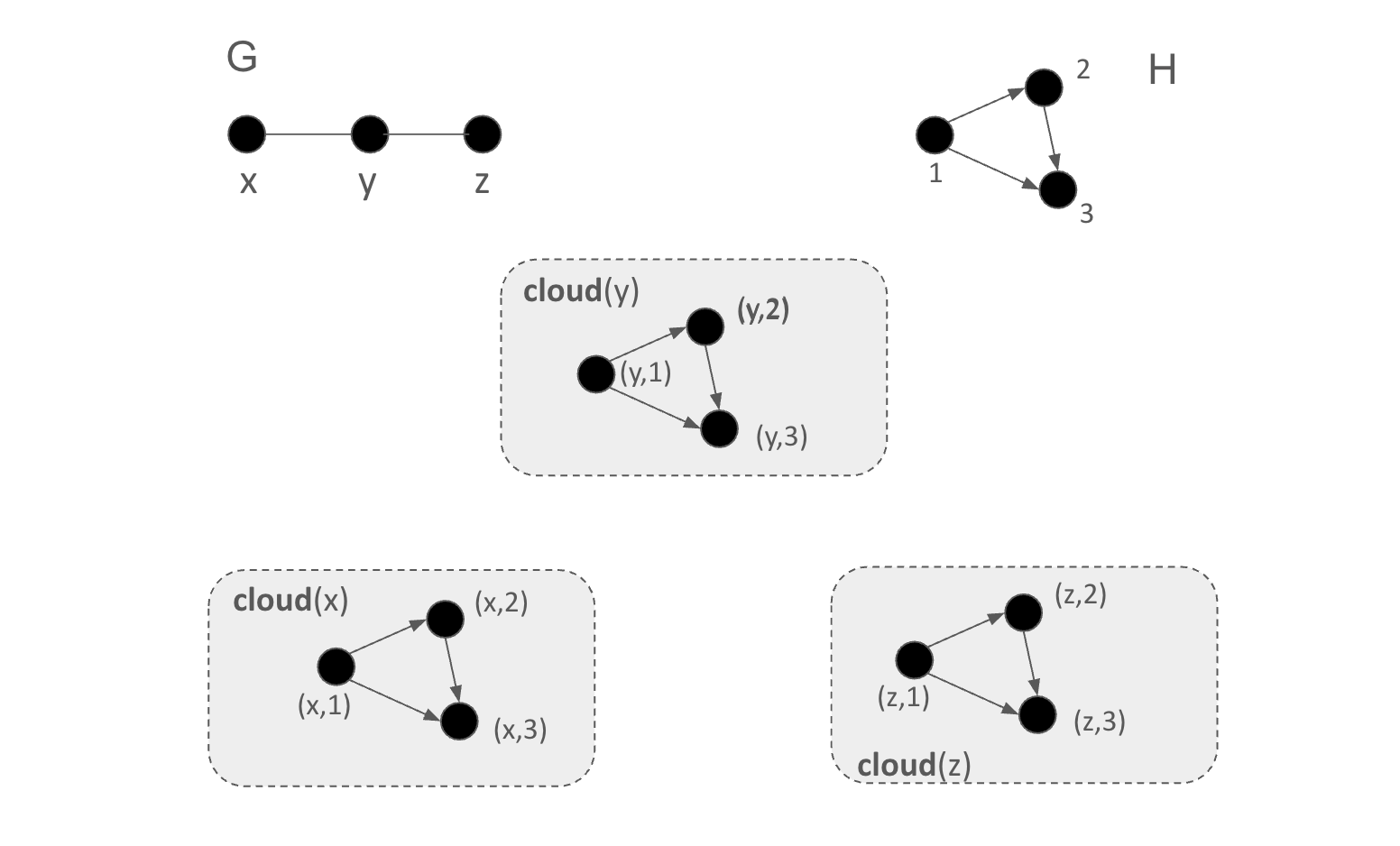}
    \caption{An illustration of $D^{\sigma}_G(H)$ after Step 1. Each vertex of $G$ is replaced by a copy of $H$. }
    \label{fig:lex1}
\end{figure}

\begin{figure}
    \centering
    \includegraphics[width=0.9\linewidth]{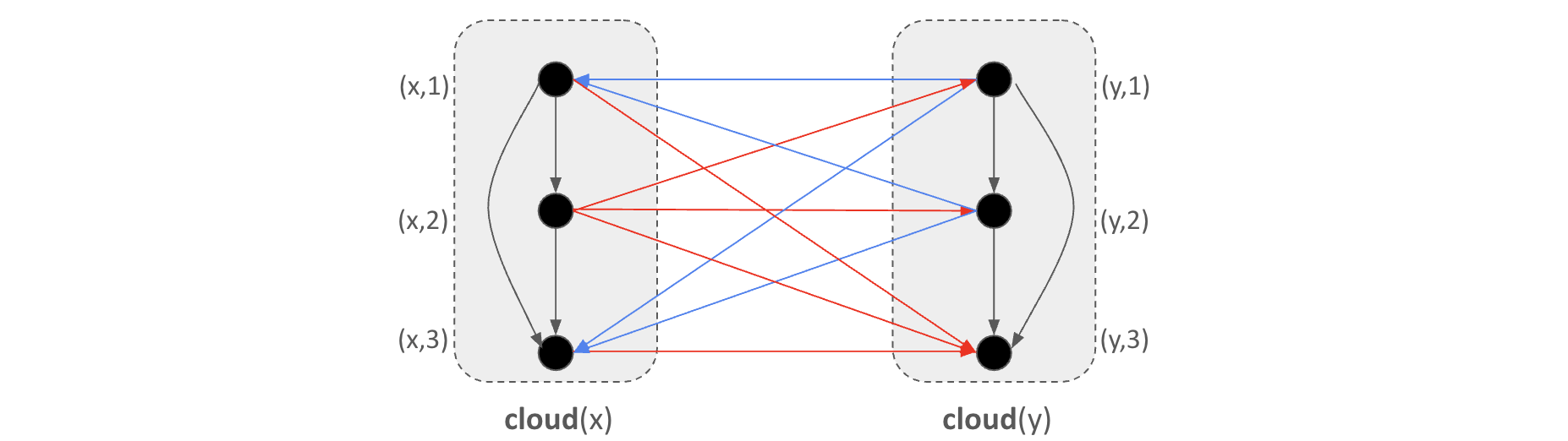}
    \caption{A (local) view of $D^{\sigma}_G(H)$ after Step 2. Arcs are added between every pair of vertices in ${\sf cloud}(x)$ and ${\sf cloud}(y)$. The directions of these arcs are chosen independently at random where blue and red arcs are oriented to the right and left respectively.  }
    \label{fig:lex2}
\end{figure}

The following property is easy to see but will be crucial for us. 

\begin{observation}
\label{obs:tournament}
The graph $D^{\sigma}_G[H]$ is a tournament if and only if $H$ is. 
\end{observation}

The following lemma is easy to see. 

\begin{lemma}
\label{lem: product chrom}
Let $H$ be a digraph and $G$ an undirected graph. Then, $\vec{\chi}(D^{\sigma}_G[H]) \leq \chi(G) \vec{\chi}(H)$ with probability one.     
\end{lemma}
\begin{proof}
Let $X_1, X_2,\ldots, X_g$ be  color classes of $G$ where $g = \chi(G)$ and $Y_1,\ldots, Y_h$ be acyclic color classes of $H$ where $h = \vec{\chi}(H)$. 
Since each $X_i$ is independent in $G$ and $Y_j$ is acyclic in $H$, we have that $X_i \times Y_j$ is also acyclic in $D^{\sigma}_G[H]$ (regardless of the orientation $\sigma$). This gives us $g \cdot h$ acyclic color classes. 
\end{proof}

Consider a sample $G' \sim D^{\sigma}_G[H]$. 
We say that edge $\{u,v\} \in E(G)$ is $\eta$-\textbf{consistent} if for every $A_u \subseteq {\sf cloud}(u)$ and $A_v \subseteq {\sf cloud}(v)$ such that $|A_u|, |A_v| \geq |V(H)|^{\eta}$, we have that $G'[A_u \cup A_v]$ contains a cycle. 
Next, we say that the sample $G' \sim D^{\sigma}_G[H]$ is $\eta$-\textbf{good} for $(G,H)$ if every edge $\{u,v\} \in E(G)$  is $\eta$-consistent. 

The following lemma argues that a random product graph is $\eta$-good with high probability. The proof closely follows~\cite{feder2019complexity} (Lemma 4.1). 

\begin{lemma}
\label{lem: good prob}
For all $\eta \in (0,1/2]$, there exists a constant $n_0$ (depending on $\eta$) such that for all $G$ and $H$ such that $|V(H)| \geq |V(G)| \geq n_0$, the probability that the graph $D^{\sigma}_G[H]$ is $\eta$-good is at least $1-\eta^2$. 
\end{lemma}

\begin{proof}
We need a simple observation: When $G'[X]$ is acyclic for subset $X \subseteq V(G')$, some topological ordering $\tau: X \rightarrow [|X|]$ is valid for $X$ (i.e., there is no edge going back in the ordering).

Fix an edge $\{u,v\} \in E(G)$. We first analyze the probability that $\{u,v\}$ is not $\eta$-consistent (denote by event $\eset_{uv}$), i.e., there exists $A_u, A_v$ such that $G'[A_u \cup A_v]$ is acyclic. 
This can be upper bounded by the probability that some topological ordering is valid for $A_u \cup A_v$. 
\begin{eqnarray*}
{\mathbb P}[\eset_{uv}] &\leq &  {\mathbb P}[(\exists A_u, A_v: |A_u|,|A_v| = |V(H)|^{\eta}) (\exists \tau)\mbox{ $\tau$ is valid for $A_u \cup A_v$}] \\ 
 &\leq & |V(H)|^{2|V(H)|^{\eta}} \cdot {\mathbb P}[(\exists \tau)\mbox{ $\tau$ is valid for $A_u \cup A_v$}] \\
 &\leq &  |V(H)|^{2|V(H)|^{\eta}}  \cdot (2|V(H)|^{\eta})! \cdot {\mathbb P}[\tau \mbox{ is valid for $A_u \cup A_v$}]
\end{eqnarray*}
Now consider a fixed $w \in A_u \cup A_v$ (say $w \in A_u$ without loss of generality). Since $w$ is adjacent to all vertices in $A_v$, the ordering $\tau$ necessitates the orientation of all such edges between $w$ and $A_v$. This means that the probability is at most $(1/2)^{|A_v|}$ for a fixed $w$, and therefore 
\[{\mathbb P}[\tau \mbox{ is valid for $A_u \cup A_v$}] \leq (1/2)^{|A_v||A_u|} \leq (1/2)^{|V(H)|^{2\eta}}\]
Plugging this term back, we get: 
\[{\mathbb P}[\eset_{uv}] \leq |V(H)|^{4|V(H)|^{\eta}} \cdot (1/2)^{|V(H)|^{2\eta}} \leq 2^{-|V(H)|^{2\eta} + 4 |V(H)|^{\eta} \lg_2 |V(H)|} \]
When $|V(H)|$ is sufficiently large, we can make ${\mathbb P}[\eset_{uv}]$ smaller than $\eta^2/|V(H)|^2$, and therefore (by the union bound over the edges of $G$) the probability that $G'$ is not $\eta$-good is at most $\eta^2$ as desired. 
\end{proof}

\subsubsection{New Graph Product Inequality} 

This subsection presents an analogue of the graph product inequality for our randomized lexicographic product.

\begin{lemma}[New inequality]
Let $G$ and $H$ be arbitrary undirected and directed graphs. Let $G' \sim D^{\sigma}_G[H]$ be such that $G'$ is  $\eta$-good for $(G,H)$. Then, we have that 
\[\vec{\alpha}(G') \leq \alpha(G) \cdot \vec{\alpha}(H) + |V(G)| |V(H)|^{\eta}.\]
\label{lem:new product}
\end{lemma}

\begin{proof}
Let $S \subseteq V(G')$ be such that $|S| = \vec{\alpha}(G')$ and $G'[S]$ is acyclic. We partition $S$ into $\bigcup_{v} S_v$ where $S_v = S \cap {\sf cloud}(v)$. We say that $S_v$ is large if $|S_v| \geq |V(H)|^{\eta}$ and small otherwise. 
Clearly, the total size of the small sets $S_v$ is at most $|V(G)| \cdot |V(H)|^{\eta}$. 

Now, for each large set $S_v$, we have that $S_v$ induces an acyclic subgraph in $H$, so $|S_v| \leq \vec{\alpha}(H)$.
Define the large clouds $L \subseteq V(G)$ where $L = \{v: S_v \mbox{ is large}\}$. It is clear that $L$ must be an independent set in $G$: Otherwise, if $u,v \in L$ has an edge $\{u,v\} \in E(G)$, this would imply that $S_u \cup S_v$ induces a cycle. Therefore, we have that the total size of large sets is at most $\alpha(G) \cdot \vec{\alpha}(H)$.  
\end{proof}

\subsubsection{Completing the Proof: The $k$-Fold Lexicographic Product}

Now we complete the proof of Lemma~\ref{lem:reduction}.  Given
undirected graph $\widehat{G}$ and a fixed ordering $\sigma$ of $V(\widehat{G})$, define a sequence of (random) 
tournaments $G_1,\ldots, G_k$ as follows: 

\begin{itemize}
    \item We construct tournament $G_1$ by orienting edges according to the ordering $\sigma$. That is, $(u,v) \in E(G_1)$ if and only if $\sigma(u) < \sigma(v)$. 
    We have that $\vec{\chi}(G_1), \vec{\alpha}(G_1) \leq |V(\widehat{G})|$. 

    \item For all $i \geq 1$, $G_{i+1} = D^{\sigma}_{\widehat{G}}[G_i]$. Since $G_1$ is a tournament, from~\Cref{obs:tournament}, graph $G_i$ must be a tournament for all $i$.    
\end{itemize}

Our final graph is $\pt= G_k$. 
Clearly, we have $|V(G_i)|= |V(\widehat{G})|^i$ and in particular $|V(\pt)| = |V(\widehat{G})|^k$. 
The following is a trivial consequence of Lemma~\ref{lem: product chrom} (applying it iteratively). 

\begin{observation}
With probability one, $\vec{\chi}(G)\leq \chi(\widehat{G})^{k-1} |V(\widehat{G})|$.    
\end{observation}

\begin{lemma}
For sufficiently large $|V(\widehat{G})|$, the probability that, for all $i  \in \{2,\ldots, k\}$, the graph $G_i$ is $(1/k)$-good for $(\widehat{G}, G_{i-1})$ is at least $(1- 1/k)$
\end{lemma}
\begin{proof}
Applying Lemma~\ref{lem: good prob} (notice that the precondition $|V(G_i)| \geq |V(\widehat{G})|$ is always satisfied), the probability that a fixed $G_i$ is not $(1/k)$-good is at most $1/k^2$. Denote by $\eset_i$ the event that $G_i$ is $(1/k)$-good for $(\widehat{G}, G_{i-1})$. 
\[{\mathbb P}[\wedge_i \eset_i]= \prod_{i}{\mathbb P}[\eset_i \mid \wedge_{j<i} \eset_j] \geq (1-1/k^2)^k \geq 1-1/k\]
\end{proof}

Now assume that all $G_i$ are $(1/k)$-good. 
From Lemma~\ref{lem:new product}, we have that, for all $i\leq k$, 
\[\vec{\alpha}(G_i) \leq \alpha(\widehat{G}) \cdot \vec{\alpha}(G_{i-1}) + |V(\widehat{G})| |V(G_{i-1})|^{1/k} \leq \alpha(\widehat{G})\cdot \vec{\alpha}(G_{i-1}) + |V(\widehat{G})|^2 \]
It is easy to derive that $\vec{\alpha}(G_k) \leq O(\alpha(\widehat{G})^k |V(\widehat{G})|^2)$, thus completing the proof of Lemma~\ref{lem:reduction}.

\setcounter{section}{2}

\section{Approximation Algorithms}\label{sec:algorithms}

In this section, we provide approximation algorithms to properly color
directed graphs promised, directly or indirectly, to have ``small''
dichromatic number.  First, we give an algorithm to color a
2-dicolorable digraph with few colors.  Then we extend this to
$\ell$-dicolorable digraphs.  Next, we consider digraphs with bounded
independence number, a direct generalization of tournaments (which
have independence number one).  Here, we can obtain bounds when we are
given additional promises such as that the input digraph is
2-dicolorable or that it is $C_3$-free.  We present some useful
observations after presenting some necessary notation.

We define $uv \in A$ to be an arc directed from $u$ to $v$.  For $v\in
V$, let $N^+(v) = \{u~|~ vu\in A\}$ and $N^-(v) = \{u~|~uv\in A\}$.
For $S \subset V$, we define $N^+(S) = \bigcup_{v \in S} N^+(v)$, and
we define $N^-(S)$ analogously.  For $S \subseteq V$, we use $D[S]$ to
denote the digraph induced on the vertex set $S$, although we
frequently abuse notation and refer to the induced subdigraph itself
as $S$.  Sometimes, we use $V(D)$ and $A(D)$ to refer to the vertex
and arc set, respectively.  We use $n$ to denote the number of
vertices in $D$, when it is clear to which digraph $D$ we are
referring.

\begin{observation}\label{O:lastVertex}
Let $D$ be a digraph such that $\chiv(D)\leq \ell$, for some
$1\leq \ell\leq n$. Then for every (induced) subgraph $D'$ of $D$,
there exists a vertex $v\in V(D')$ such that $\chiv(D'[N^+(v)]) \leq
\ell-1$.
\end{observation}
 \begin{proof}
Since $\chiv(D)\leq \ell$, there is a partition of the vertices into
$V(D) = \bigcup_{1\leq i \leq \ell} S_i$ such that each set $S_i$ is
acyclic.  Now consider an ordering of $V(D)$ which consists of an
acyclic ordering of $S_1$ (i.e., with no backward arcs from $A(S_1)$)
followed by an acyclic ordering of $S_2$, and so on.  Notice that the
last vertex in this ordering $v$ has $N^+(v) \subseteq \bigcup_{1\leq
  j\leq \ell-1} S_j$, which is $(\ell-1)$-colorable.
\end{proof}

Given a digraph $D$ and a total order $\prec$ on $V(D)$, let $\BG$
denote the undirected graph whose edges correspond to arcs of $D$ that
are backwards with respect to $\prec$ (i.e., by the arcs $ vu \in
A(D)$ such that $u\prec v$).  In other words, we consider all
backwards arcs with respect to the ordering $\prec$ and then we ignore
the orientation of the arcs, resulting in an undirected {\em backedge
  graph}.

\begin{observation}[\cite{nguyen2025some}]\label{O:backedge}
Let $\BG$ be a backedge graph associated with $D$ and $\prec$. Then
$\chiv(D) \leq \chi(\BG)$. Moreover, any proper coloring of vertices
of $\BG$ is also a proper coloring for $D$.
\end{observation}

This follows from the fact that every stable set in $\BG$ is acyclic
in $D$. Recall that {\em $d$-degenerate} (undirected) graph is one in
which there is a vertex ordering such that each vertex has at most $d$
neighbors to the left.

\begin{observation}[Folklore]\label{P:degenerate}
A $d$-degenerate graph is $(d+1)$-colorable and such a coloring can be computed in polynomial time.
\end{observation}

A digraph $D$ is $d$-\textit{out-degenerate} if every induced subgraph of $D$ contains a vertex with out-degree at most $d$. An implication of Observation \ref{O:backedge} and \ref{P:degenerate} is the following.

\begin{observation}
    A $d$-out-degenerate digraph is $(d+1)$-colorable and such a coloring can be computed in polynomial time.
\end{observation}

We remark that we will always assume that a digraph $D$ that we want to color is strongly connected; if this were not the case, we can color each strongly connected component separately, using the same color palette.

\subsection{Coloring $\ell$-Dicolorable Digraphs}

We now consider the case in which we are promised an upper bound on
the dichromatic number of an input digraph and the goal is to color it
with few colors.  We first consider the case in which the input digraph $D$ is promised to be 2-colorable.  Recall that coloring such a digraph with two colors, even when it is oriented, is \NP-hard~\cite{bokal2004circular}.  
We present Algorithm \ref{alg:2coloring} and
use it to prove Theorem \ref{T:2-approximate}.

\begin{restatable}{theorem}{twoApproximate}\label{T:2-approximate}
Given a digraph $D$ with $\vec{\chi}(D) \leq 2$, there exists a
polynomial-time algorithm to properly color $D$ with at most
$2\sqrt{n}$ colors.
\end{restatable}

Algorithm \ref{alg:2coloring} has a similarity to Wigderson's algorithm for coloring 3-colorable graphs~\cite{wigderson1983improving}.  However, one key difference between coloring 3-colorable graphs and coloring 2-dicolorable digraphs is that in the case of the former, each vertex has a neighborhood that is easy to color, since each vertex neighborhood is bipartite.  In the latter case, we are only guaranteed the existence of two vertices, whose out-neighborhoods are acyclic and are therefore easy to color, and both of these vertices might have small out-degree, preventing us from making progress.  We circumvent this issue by finding another way to remove small-degree vertices, using Observation \ref{O:backedge}.

\vspace{3mm}

\begin{algorithm}[H]
\caption{Dicoloring($D$)}
\label{alg:2coloring}
Input: {A digraph $D$ on $n$ vertices with $\chiv(D)\leq 2$.}

Output: A coloring of $V(D)$ using at most $2\sqrt{n}$ colors.
\begin{enumerate}
\item Initialize $i=1$, set $S=\emptyset$, and an order $\prec$ on $S$.
\item While $(V(D))$:
\begin{enumerate}
\item Let $v\in V(D)$ be a vertex such that $D[N^+(v)]$ is acyclic. (Observation~\ref{O:lastVertex})

\item If $|N^+(v)|\geq \sqrt{n}$, then:
\begin{enumerate}
    \item Assign color $i$ to each vertex in $N^+(v)$.
    \item Set $i := i+1$ and $D := D[V(D)\setminus N^+(v)]$.
\end{enumerate}

 \item Else if $|N^+(v)|< \sqrt{n}$:
 \begin{enumerate}
     \item Add $v$ to $S$ and update $\prec$ such that $v\prec u$ for each $u\in S\setminus \{v\}$.
     \item  $D := D[V(D)\setminus \{v\}]$.
 \end{enumerate}

\end{enumerate}
\item Let $\BG$ be the backedge graph corresponding to $D[S]$ and the order $\prec$.
\item Color $\BG$ using $\sqrt{n}$ new colors from $\{\sqrt{n}, \ldots, 2\sqrt{n}\}$. (Observation~\ref{P:degenerate}) 
\item If a vertex $v\in V(\BG)$ receives color $j$, then assign color $j$ to $v$ in $D$. (Observation~\ref{O:backedge})
\end{enumerate}

\end{algorithm}

\begin{proof}[Proof of Theorem \ref{T:2-approximate}]
We use the following claims to prove the correctness of Algorithm \ref{alg:2coloring}.
\begin{claim}\label{O:trivial}
    In Step~2 of Algorithm \ref{alg:2coloring}, we assign at most $\sqrt{n}$ colors from colors $\{1,\ldots,\sqrt{n}\}$.
\end{claim}
\begin{cproof}
    Since each time we use a new color we remove at least $\sqrt{n}$ vertices, we assign at most $\sqrt{n}$ colors in Step~2.
\end{cproof}

\begin{claim}
    The coloring assigned by Algorithm~\ref{alg:2coloring} is a proper coloring.
\end{claim}
\begin{cproof}
First, observe that the backedge graph $\BG$ corresponding to $D[S]$ and $\prec$ is ($\sqrt{n}-1$)-degenerate. Hence we can properly color $G$ using $\sqrt{n}$ colors from $\{\sqrt{n}, \ldots, 2\sqrt{n}\}$ 
\footnote{Notice that if we actually assign $\sqrt{n}$ colors in Step~2, then we have colored the whole graph.  So if we reach Step~4, we can begin with color $\sqrt{n}$.}
due to Observation~\ref{P:degenerate}. Moreover, due to Observation~\ref{O:backedge}, this coloring is a proper coloring for $D[S]$. Finally, since  we assign a unique color from $\{1,\ldots,\sqrt{n}\}$ (distinct from colors assigned to vertices of $D[S]$) to an acyclic subdigraph of $D$ in each iteration of Step~2, this directly extends to a proper coloring of $D$. 
\end{cproof}

Since we assign at most $\sqrt{n}$ colors in Step~2 and at most $\sqrt{n}$ colors in Step~3--5 (Observation~\ref{P:degenerate}), we use a total of at most $2\sqrt{n}$ many colors to color $D$.
\end{proof}

We remark that Algorithm \ref{alg:2coloring} can easily be modified so that, on an input digraph $D$ that is not 2-dicolorable, it either outputs a coloring of $D$ with at most $2\sqrt{n}$ colors or outputs an induced subgraph of $D$ in which no vertex has an acyclic out-neighborhood, which is a proof that $D$ is not 2-dicolorable.

\vspace{5mm}

A positive answer to the following question would immediately yield an improved bound on the number of colors needed to color a 2-colorable digraph.

\begin{restatable}{question}{boundedMaxDegree}\label{KMSQuestion}
Given a 2-dicolorable $d$-out-degenerate digraph $D$, can we find a coloring of $D$ using at most $d^{1-\epsilon}$ colors in polynomial time?
   \end{restatable}

Next, we provide a $O(n^{2\ell})$-time algorithm to color a digraph $D$
with $\chiv(D) \leq \ell$ using at most $O(n^{1-\frac{1}{\ell}})$
colors, which we use to prove Theorem \ref{T:l-approximate}, of which
Theorem \ref{T:2-approximate} is a special case.

\vspace{3mm}

\begin{algorithm}[H]
\caption{DicoloringGeneral($D$, $\ell+1$)}
\label{alg:lcoloring}
Input: {A digraph $D$ on $n$ vertices with $\chiv(D)\leq \ell+1$.}

Output: A coloring of $V(D)$ using $O(n^{1-\frac{1}{\ell+1}})$ colors.
\begin{enumerate}
\item Initialize $i=0$, set $S=\emptyset$, and an order $\prec$ on $S$.
\item While $(V(D))$:
\begin{enumerate}

\item While there is a vertex $v\in V(D)$ such that $|N^+(v)| < n^{1-\frac{1}{\ell+1}}$, then:
\begin{enumerate}
     \item Add $v$ to $S$ and update $\prec$ such that $v\prec u$ for each $u\in S\setminus \{v\}$.
     \item  $D := D[V(D)\setminus \{v\}]$.
 \end{enumerate}
\item For vertex $v\in V(D)$, run
  DicoloringGeneral($D[N^+(v)]$, $\ell$) until a successful vertex $v$ is
  found (i.e., $D[N^+(v)]$ is colored with $\ell\cdot
  n^{1-\frac{1}{\ell}}$ colors). (Observation~\ref{O:lastVertex})

\item If no ``successful'' vertex is found, return {\bf FAILURE}.

\item If the algorithm colors $N^+(v)$ in Step (b) using $p$
  colors, where $v$ is a ``successful'' vertex, then:
\begin{enumerate}
    \item If a vertex $u\in N^+(v)$ is assigned color $b$, then  assign color $i+b$ to $u$ in $D$.
    \item Set $i := i+p$ and $D := D[V(D)\setminus N^+(v)]$.
\end{enumerate}

\end{enumerate}
\item Let $\BG$ be the backedge graph corresponding to $D[S]$ and the order $\prec$.
\item Color $\BG$ using $n^{1-\frac{1}{\ell+1}}$ new colors from $\{i+1, \ldots, i+n^{1-\frac{1}{\ell+1}} \}$. (Observation~\ref{P:degenerate})
\item If a vertex $v\in V(\BG)$ receives color $j$, then assign color $j$ to $v$ in $D$. (Observation~\ref{O:backedge})
\end{enumerate}

\end{algorithm}

\begin{restatable}{theorem}{lApproximate}\label{T:l-approximate}
Given a digraph $D$ with $\chiv(D) \leq \ell$, there exists an
algorithm to properly color $D$ with at most $\ell\cdot
n^{1-\frac{1}{\ell}}$ colors in time $O(n^{2\ell})$.
\end{restatable}

\begin{proof}
We will prove Theorem \ref{T:l-approximate} using
Algorithm~\ref{alg:lcoloring}.  Let $D$ be a digraph such that
$\vec{\chi}(D) \leq \ell+1$.  Due to \Cref{O:lastVertex}, there exists
at least one vertex $v\in V(D)$ such that $\vec{\chi}(N^+(v))\leq
\ell$.  Thus, each time we run Step 2, Step (b) will find a successful
vertex.  Let us assume by induction that DicoloringGeneral($D'$,
$\ell$) returns a coloring of $D'$ using at most $\ell \cdot
V(D')^{1-\frac{1}{\ell}}$ colors in $O(V(D')^{2\ell})$ time when
$\chiv(D') \leq \ell$.  Notice that this statement is true in the base
case when $\ell=2$ by Theorem \ref{T:2-approximate}.  We use the
following claim.

\begin{claim}\label{C:lcolor}
The maximum value of $i$ in an execution of Algorithm~\ref{alg:lcoloring} is at most $\ell\cdot n^{1-\frac{1}{\ell+1}}$.
    \end{claim}
    \begin{cproof}
Observe that the value of $i$ increases only in Step~2(d) of Algorithm
\ref{alg:lcoloring}, and this occurs when we color some (at least
$n^{1-\frac{1}{\ell+1}}$) vertices of $D$.  Suppose this process occurs $q$ times; that is, there are $q$ disjoint sets of vertices
$S_1,\ldots, S_q$ such that $S_j$ is colored in the $j$'th execution of
Step~2(b) of Algorithm \ref{alg:lcoloring} using at most $\ell\cdot
|S_j|^{1-\frac{1}{\ell}} = \ell\cdot|S_j|^{\frac{\ell-1}{\ell}}$
colors.  Then we have
\begin{eqnarray*}
q \cdot n^{1-\frac{1}{\ell+1}} \leq_{j \in [q]} |S_j| \leq n \quad
\Rightarrow \quad q \leq n^{\frac{1}{\ell+1}}.
\end{eqnarray*}

At the end of the algorithm, $i = \sum_{j\in [q]} \ell\cdot
|S_j|^{\frac{\ell-1}{\ell}}$. 
To complete the proof, we need to establish that $\sum_{j\in [q]}
\ell\cdot |S_j|^{\frac{\ell-1}{\ell}} \leq \ell\cdot
n^{1-\frac{1}{\ell+1}}$.  
Letting $n_j = |S_j|$ and noting that 
the function $f(x) = x^{\frac{\ell-1}{\ell}}$ is concave,
we have by Jensen's inequality,
\begin{eqnarray*}
\sum_{j\in [q]} n_j^{\left( \frac{\ell-1}{\ell} \right)} \leq q \bigg(
\frac 1 q\sum_{j\in [q]} n_j \bigg)^{\frac{\ell-1}{\ell}} \leq 
q \left( \frac{n}{q}\right)^{\frac{\ell-1}{\ell}} =
q^{\frac{1}{\ell}} n^{\frac{\ell-1}{\ell}} \leq 
(n^{\frac{1}{\ell+1}})^{\frac{1}{\ell}} n^{\frac{\ell-1}{\ell}} =
n^{\frac{\ell}{\ell+1}}
= n^{1-\frac{1}{\ell+1}}.
\end{eqnarray*}
This completes the proof of the claim.
    \end{cproof}

Thus, due to \Cref{C:lcolor}, we have that the total number of colors
used in Step~2 of Algorithm \ref{alg:lcoloring} is at most $\ell\cdot
n^{1-\frac{1}{\ell+1}}$. Since we use at most $n^{1-\frac{1}{\ell+1}}$
colors in Step~4 and 5 of the algorithm, we have that the total number
of colors used by Algorithm \ref{alg:lcoloring} is at most
$(\ell+1)\cdot n^{1-\frac{1}{\ell+1}}$. Finally, the running time
follows from the fact that after at most $n$ invocations of the
subroutine to color $\ell$-colorable digraphs using at most $\ell\cdot
n^{1-\frac{1}{\ell}}$ colors, we remove the (non-empty) out-neighborhood of at
least one vertex from the digraph $D$.  Thus, we call the subroutine
DicolorGeneral($D'$, $\ell$) at most $n^2$ times for a total running
time of $O(n^2 n^{2\ell}) = O(n^{2(\ell+1)})$.
\end{proof}

\subsection{Coloring Digraphs with Bounded Independence Number}

We now consider digraphs which have bounded independence number.  The
{\em independence number} of a digraph is the size of the maximum
independent set of its underlying undirected graph.  Tournaments are a
well-studied class of digraphs with independence number one.  Digraphs
with bounded independence number are sometimes referred to as {\em
  dense digraphs}, since for fixed independence number $\alpha$, they
have a quadratic number of edges. 
One can hope that some of the algorithmic results for tournaments can
be extended to general digraphs parametrized by independence number.
In this section, we show that this is the case (i) when we consider
digraphs promised to be 2-dicolorable, and (ii) when we consider
oriented digraphs with no $C_3$, where $C_k$ is a cyclically oriented
cycle on $k$ vertices.
The latter result improves upon a bound given
by \cite{harutyunyan2019coloring} via a simpler algorithm.

We refer to the set of vertices that have no arc from or towards $v$
as its \emph{non-neighborhood} $N^o(v)$ (i.e., $N^o(v) = V\setminus
\{N^+(v) \cup N^-(v) \cup v\}$).  We say there is a \emph{non-edge}
between $u$ and $v$ if there is no arc in $A$ from $u$ to $v$ or from
$v$ to $u$.  If we are referring to an non-edge with a fixed direction
(i.e., a non-edge $e = uv$), then we refer to it as a \emph{non-arc}.
Notice that if $uv$ is a non-arc, then there is a non-edge between $u$
and $v$ (i.e., if $uv$ is an arc, then we do not say that $vu$ is a
non-arc).  For an arc or non-arc $e=uv$, we use $N(e)$ to denote the
vertices in $N^+(v) \cap N^-(u)$.

The following observation is useful in digraphs with bounded independence number.  

\begin{observation}\label{obs:alphaMinusOne}
Let $D$ be a digraph with independence number at most $\alpha$. Then for
every vertex $v \in V(D)$, the digraph $D[N^o(v)]$ has independence number at most $\alpha -1$.
\end{observation}

One of the main tools presented in \cite{klingelhoefer2024coloring} for coloring tournaments is a {\em path decomposition}, in which a tournament is decomposed into vertex sets, which are essentially the neighborhoods of arcs on a shortest path between two vertices $s$ and $t$, where $N(e)$ (defined above) is the neighborhood of arc $e$.  Moreover, if $s$ and $t$ are chosen such that $N^+(s)$ and $N^-(t)$ each have small dichromatic number and $N(e)$ has small dichromatic number for each arc $e$ on the shortest path, then this decomposition can be used to show that the whole tournament has small dichromatic number.  The main idea is that we can reuse color palettes for neighborhoods of arcs that are far apart on the shortest path, because arcs between these neighborhoods must go backwards, otherwise it would contradict the minimality of the shortest path.  

In this section, we adapt this path decomposition to oriented digraphs by taking a shortest path---not just of arcs---but of arcs and non-arcs (in other words, non-edges can be included in the shortest path to be traversed in
either direction).  Thus, we can fix some pair of vertices $s$ and $t$, and then we can partition the vertex set of a digraph into neighborhoods of arcs on the shortest path and non-neighborhoods of vertices on the shortest path.  We can still ensure that arcs between neighborhoods of arcs and arcs between non-neighborhoods of vertices that are far apart on the shortest path must go backwards, allowing us to reuse a relatively small color palette to color all the vertices of the digraph, except $N^+(s)$ and $N^-(t)$, which must therefore be chosen carefully.

Hence, to apply the path decomposition, the main challenges are finding a good vertex pair $s,t$ such that $N^+(s)$ and $N^-(t)$ each have small dichromatic number, and showing that the non-neighborhoods of the vertices and the neighborhoods of the arcs or non-arcs on a shortest path from $s$ to $t$ each have small dichromatic number.  However, we can define the decomposition for any $s,t$ vertex pair.

\subsubsection{Decomposition for Digraphs}

In this section, we define a decomposition for digraphs based on a shortest path from $s$ to $t$ for any vertex pair $s$ and $t$.  In the next definition, we have $v_0 := s$ and $v_k := t$.

\begin{definition}\label{def:vc_gen_dig}
We define a \emph{vertex chain} $(v_i)_{0\leq i \leq k}$ of a digraph
$D$ as follows: Let $v_0$ and $v_k$ be a pair of vertices and let
$(v_i)_{0\leq i \leq k}$ be the vertices in the shortest path from
$v_0$ to $v_k$, where the path may consist of both forward arcs and
non-arcs (but no backward arcs).  If in addition, $\chiv(N^+(v_0))
\leq b$ and $\chiv(N^-(v_k)) \leq b$, then we call it a
\emph{$b$-vertex chain}.
\end{definition}

We remark that in the above definition, $v_0$ and $v_k$ are not
necessarily distinct, in which case the vertex chain consists of
a single vertex and the shortest path has length zero.

Additionally, we define an \emph{edge chain} $(e_i)_{1 \leq i \leq k}$
corresponding to a vertex chain, where $e_i$ is the arc or non-arc
from $v_{i-1}$ to $v_i$.  We build zones that can be efficiently
colored, and such that arcs between zones at distance more than four
(i.e., \textit{long} arcs) go backwards.

\begin{definition}\label{def:edge_dig} Given a vertex chain $(v_i)_{0\leq i \leq k}$, a \ndecomp of a digraph $D$ is defined as:
	\begin{itemize}
	\item $D_1 = N(e_1)$.
		\item For $2 \leq i \leq k$, $D_i = N(e_i) \setminus (\cup_{1\leq j \leq i-1} D_j)$.
		\item For $0 \leq i \leq k$, $N_i = N^o(v_i)\setminus{(\cup_{1\leq j \leq k} D_j \cup_{0 \leq j \leq i-1} N_j)}$. 
		\item $D_0 = N^+(v_0)\setminus (\cup_{1\leq j \leq k} (D_j \cup N_j)\cup N_0 )$.
		\item $D_{k+1} = N^-(v_k) \setminus \cup_{0\leq j \leq
                  k} (D_j \cup N_j)$.
\item $Z = V \setminus (\cup_{0\leq j \leq
                  k} (D_j \cup N_j) \cup D_{k+1})$.
	\end{itemize}
\end{definition}

First we prove that this is indeed a decomposition of $D$.

\begin{lemma}\label{clm:decomp_dig}
Let $D=(V,A)$ be a digraph, and $(D_0, \ldots, D_{k+1}, N_0, \ldots,N_k, Z)$ be a \ndecomp of $D$.  Then $Z \subseteq \{v_0,v_k\}$.
\end{lemma}

\begin{proof}
We will prove this by contradiction. Suppose there is a vertex $w \in
V \setminus \{v_i\}_{0 \le i \le k}$ that does not belong to any $D_i$ or $N_i$.  Since $w$ does not
belong to $D_0$ or to $D_{k+1}$ (nor to $N_0$ or $N_k$), then $w \in
N^-(v_0)$ and $w \in N^+(v_k)$.  Take the smallest integer $i$ such
that $w \in N^+(v_i)$.  There must be one since $w \in N^+(v_k)$.
Notice that $i \geq 1$ since $w \in N^-(v_0)$.  Since $w \notin
N^o(v_{i-1})$, then $w \in N(e_i)$.  Therefore, $w \in D_i$, which is
a contradiction.

We note that each vertex in the vertex chain $(v_i)_{0 \leq i \leq
  k}$ belongs to a $D_i$ or $N_i$, except in the case where $k=0$ or
$k=1$.  In this case, we add this set of at most two remaining vertices to $Z$.  Finally, we remark that the decomposition is actually a partition of $V$.
\end{proof}

We now show that long arcs between $N_i$'s or between $D_i$'s go
backwards.  Since the $N_i$s and $D_i$s will be colored with different
color palettes, we do not need to worry about arcs between an $N_i$
and a $D_j$.

\begin{lemma}\label{clm:long_arc_dig}
Let $0 \leq i,j \leq k+1$ such that $j \geq i+5$.  For $u\in D_i$ and
$w\in D_j$, we have $u\in N^+(w)$.  Furthermore, if $u \in N_i$ and $w
\in N_{j-1}$, we have $u\in N^+(w)$.
\end{lemma}

\begin{proof}
We will prove this by contradiction. Suppose $j \geq i+5$ and $u \in
N^-(w)$ (respectively, $u\in N^o(w)$).  Then there is a path of three
arcs (respectively, of two arcs and one non-arc) from $v_i$ to
$v_{j-1}$, namely $(v_i, u, w, v_{j-1})$.  (By definition of the
decomposition, $u \in D_i$ implies $u \in N^+(v_i)$ and $w \in D_j$
implies $w \in N^-(v_{j-1})$.) This is not possible since by the
definition of the vertex chain as the shortest path, there can be no
path between $v_i$ and $v_{j-1}$ with fewer than four arcs (since
$(j-1)-i \geq (i+5 -1)-i = 4$).  Finally, if $u \in N_i$ and $w \in
N_{j-1}$, $(v_i, u, w, v_{j-1})$ is also a path (of forward arcs and
non-arcs) of length three from $v_i$ to $v_{j-1}$, which as previously
is a contradiction.
\end{proof}

The following lemma demonstrates how the vertex chain and path decomposition
can be used to color digraphs.

\begin{lemma}\label{lem:eff_loc_to_glob_dig}
If $D$ has a $b$-vertex chain that can be found in polynomial time and
if $\chiv(N(e)) \leq c$ for each arc and non-arc $e$, and
$\chiv(N^o(v)) \leq d$ for every vertex $v$, then
\[
\chiv(D) \leq 
\begin{cases}
5c + 4d + 2(b - c) & \text{if } b > c, \\
5c + 4d & \text{if } b \leq c.
\end{cases}
\]

\noindent Furthermore, if there are no arcs from $D_0$ to $D_{k+1}$ and $\max\{b,c,d\} = b$, then $\chiv(D) \leq 4c + 4d + b$.
\end{lemma}

\begin{proof}
Given a $b$-vertex chain, we construct a \ndecomp as per Definition
\ref{def:vc_gen_dig}.  We make five palettes of $c$ colors each with
labels from $0$ to $4$.  We color each $D_i$ using the color palette
with label $i \bmod 5$.  Note that the sets $D_0$ and $D_{k+1}$ each
require $b-c$ extra colors in their palettes if $b > c$.

Next, we make four different palettes of $d$ colors each with labels
from $5$ to $8$.  We color each $N_i$ using the color palette with
label $5+ (i \bmod 4)$.  The set of colors used is of size $d$ for
every $N_i$, thus we can efficiently color each set by the assumption
of the lemma.

Our goal is now to prove that this is a proper coloring of the digraph
$D$.  We will do this by showing that all forward arcs between
different $D_i$ or $N_i$ are bicolored.  By Lemma
\ref{clm:long_arc_dig}, there are no forward arcs between $D_i$ and
$D_j$ when $j\geq i+5$, or $N_i$ and $N_j$ with $j \geq i+4$.
Furthermore, by the definition of the coloring, no vertex in $D_i$ and
$D_j$ can share a color for $i+1 \leq j \leq i+4$, and the same goes
for vertices in $N_i$ and $N_j$ with $i+1 \leq j \leq i+3$.  Thus all
forward arcs from $D_i$ to $D_j$ or $N_i$ to $N_j$ will be bicolored.
Since every $D_i$ and $N_i$ is properly colored, and all forward arcs
between different $D_i$ are bicolored, the unions of $D_i$'s and the
union of $N_i$'s are both properly colored.  Since these two sets use
disjoint color palettes, and because every vertex is in some $D_i$ or
$N_i$ (by Lemma \ref{clm:decomp_dig}), the whole digraph $D$ is
properly colored.  Furthermore, we note that if there are no arcs from
$D_0$ to $D_{k+1}$, then we can save $b-c$ colors resulting in the
claimed bound.

Finally, we remark that if the vertex chain consists of at most two
vertices, then we have $\chiv(D) \leq 2d + c + 2b + 2$, since we can
decompose $D$ into (possibly some subset of) $Z, N_0, N_1, D_0, D_k$
and $D_{k+1}$.  In the case where there are no arcs from $D_0$ to
$D_{k+1}$, we have $\chiv(D) \leq 2d +c +2+ b$.
\end{proof}

\subsubsection{Coloring 2-Dicolorable Dense Digraphs}

In this section, our goal is to prove Theorem \ref{thm:2-col_dig}.

\begin{theorem}\label{thm:2-col_dig}
Let $D$ be a 2-dicolorable digraph with independence number
$\alpha$. Then a coloring with at most $\frac{10}{3}(4^\alpha-1)$
colors can be found in polynomial time.
\end{theorem}

We say an arc $e$ in $D$ is \emph{heavy} if $N(e)$ contains a directed
cycle.  Moreover, we say a non-edge between $u$ and $v$ is
\emph{heavy} if $N^+(u) \cap N^-(v)$ or if $N^+(v) \cap N^-(u)$
contains a directed cycle.  If a digraph contains no heavy arcs and no
heavy non-edges, then it is {\em light}.  In the case of 2-colorable
tournaments, we can partition the vertex set into two tournaments such
that no arc is heavy~\cite{klingelhoefer2024coloring}, but in the case
of 2-colorable digraphs we still need to account for the non-arcs.  To
do this, we will add a step before this partition.

\begin{lemma}\label{lem:heavy_non_arcs}
A 2-colorable digraph $D$ can be transformed into a 2-colorable
digraph $D'$ with no heavy non-arc by adding arcs.
\end{lemma}

\begin{proof}
Let $e$ be a heavy non-arc between $u$ and $v$.  Without loss of
generality, suppose that there is a directed cycle in $N^+(u) \cap
N^-(v)$. This cycle must be colored with two colors, thus $u$, $v$ and
any $z \in N^-(u) \cap N^+(v)$ cannot all be colored the same color,
else there would be a monochromatic 4-cycle. Therefore, we can add the
arc from $u$ to $v$, since it will not lead to any monochromatic cycle
in a 2-coloring of $D$.  So we can construct $D'$ by starting from
$D$, and while there are still heavy non-arcs, we add an arc between
two vertices in the way we just described.
\end{proof}

After applying Lemma \ref{lem:heavy_non_arcs}, we obtain a digraph in
which every non-arc is light.  (Of course, we have not increased the
independence number, since we have not removed any arcs.) If a digraph
contains a digon (i.e., a directed 2-cycle), then both arcs can be
considered to be heavy, since the endpoints must have different
colors.  Since the endpoints of heavy arcs must have different colors,
they cannot contain any odd cycles in a 2-coloring.
In other words, the set of heavy arcs forms a
bipartite graph and we can therefore partition the set of vertices in
order to cut all heavy arcs. This results in the following corollary.

\begin{corollary}\label{cor:dig_split}
Let $D$ be a 2-colorable digraph.  Then $D$ can be partitioned into
two light oriented 2-colorable digraphs $D_1$ and $D_2$
such that $\chiv(D) \leq \chiv(D_1) + \chiv(D_2)$.
\end{corollary}
Since finding this partition is equivalent to finding a bipartition of
a bipartite graph, it can be found in polynomial time.
Next, we prove another useful observation.

\begin{observation}\label{clm:min_max_vertex-k_dig}
Let $D$ be a $k$-colorable digraph.  Then there exist vertices $u$ and
$w$ such that $D[N^+(u) \cup N^-(w)]$ is $(k-1)$-colorable.
\end{observation}

\begin{proof}
Since $D=(V,A)$ is $k$-colorable, there exist $k$ transitive sets
$X_1, \ldots, X_k$ such that $V= \cup_{i=1}^k X_i$. Then take $u$ to
be the vertex in $X_1$ that has only incoming arcs (or non-arcs) from
other vertices in $X_1$ (i.e., the sink vertex for $X_1$). Similarly,
take $w$ to be the vertex in $X_1$ that has only outgoing arcs (or
non-arcs) to other vertices in $X_1$ (i.e., the source vertex for
$X_1$).  The out-neighborhood of $u$ and the in-neighborhood of $w$
are both subsets of $V \setminus{X_1}$, and thus so is their union,
which is therefore $(k-1)$-colorable.
\end{proof}

Now we are ready to prove Theorem \ref{thm:2-col_dig}.  We use a
subroutine for coloring a light, 2-colorable digraph
whose steps are shown in Algorithm~\ref{alg:2AlphaColoringLight}.

\begin{algorithm}
\caption{DicoloringLightIND($D,\alpha,P$)}
\label{alg:2AlphaColoringLight}
Input: {A light 2-dicolorable digraph $D$ with independence number at most $\alpha$ and a palette $P$ of $\frac{5}{3}(4^\alpha-1)$
colors.}

Output: A coloring of $V(D)$ using palette $P$. 
\begin{enumerate}
\item If $D$ is a tournament, color $D$ using 5 colors from $P$ and return. (\cite{klingelhoefer2024coloring})
\item Partition $P$ arbitrarily into five palettes $P_0,\ldots,P_4$ such that for $0\leq i \leq 3$, $|P_i| = \frac{5}{3}(4^{\alpha-1}-1)$ and  $|P_4|= 5$.
\item Find $s,t \in V(D)$ such that   $N^+(s) \cup N^-(t)$ is acyclic. (\Cref{clm:min_max_vertex-k_dig})
\item Find a $1$-vertex chain $(v_i)_{0\leq i \leq k}$ where $v_0=s$ and $v_k=t$.
\item For $0\leq i\leq k+1$, color $\bigcup D_i$ using five colors from $P_4$. (\Cref{lem:eff_loc_to_glob_dig})  
\item For $0\leq i\leq k$:
\begin{enumerate}

\item If $N_i \neq \emptyset$: DicoloringLightIND($D[N_i],\alpha-1, P_{i \mod 4}$).

\end{enumerate}

\end{enumerate}

\end{algorithm}

\begin{proof}[Proof of Theorem \ref{thm:2-col_dig}]
We begin by using \Cref{lem:heavy_non_arcs} to transform the input
digraph $D$ into a digraph containing no heavy non-arcs.  Next, if $D$
is not a light digraph, we partition $D$ into two light $2$-colorable
digraphs $D_1,D_2$ using \Cref{cor:dig_split}.  

Then we color light digraphs $D_1$ and $D_2$ separately using
different color palettes.  To color, say, $D_1$, we use a recursive
algorithm, whose steps are shown in
Algorithm~\ref{alg:2AlphaColoringLight}.  (Notice that this
presentation includes some technical implementation details on how to
distribute the color palettes during the recursive steps.)  First, we
find a $1$-vertex-chain in $D_1$.  To do this, we apply Observation
\ref{clm:min_max_vertex-k_dig}, which takes at most $O(n^2)$ time
(where $n = |V(D_1)|$), by guessing all pairs of endpoints $s,t$, and
checking if the subgraph $D_1[N^+(s) \cup N^-(t)]$ is acyclic.  Once
we have $s$ and $t$, we find a 1-vertex chain $(v_i)_{0 \leq i \leq
  k}$ where $v_0=s$ and $v_k=t$.  Since $D_1$ is light, we have
$\chiv(e_i) \leq 1$ for each arc and each non-arc $e_i$ (since each
non-arc is light) in the corresponding edge-chain.  Moreover, the
independence number of each $N_i$ is at most $\alpha-1$ by Observation
\ref{obs:alphaMinusOne}.  

In order to prove the upper bound on the number of colors stated in
the theorem, we define $f(\alpha)$ to be a function such that
$f(\alpha) \geq \chiv(D)$ for every 2-colorable digraph $D$ with
independence number $\alpha$.  Our goal is to find an upper bound on
$f(\alpha)$.  Additionally, we define $\fl(\alpha)$ to be a function
such that $\fl(\alpha) \geq \chiv(D)$ for every 2-colorable light
digraph.  By Corollary, \ref{cor:dig_split} we have $f(\alpha) \leq 2
\cdot \fl(\alpha)$.  By Lemma \ref{lem:eff_loc_to_glob_dig}, we have
the following claim.

\begin{claim}
$\fl(\alpha) = 5 + 4 \cdot \fl(\alpha - 1)$.
  \end{claim}

The base case is $\fl(1) = 5$, by \cite{klingelhoefer2024coloring}.
A simple calculation leads to the upper bound in the statement of
Theorem \ref{thm:2-col_dig}.

It remains to show that Algorithm~\ref{alg:2AlphaColoringLight} runs
in polynomial time in $n$ (the size of $V(D_1)$).  As noted above, we
can find a 1-vertex chain and a corresponding path decomposition in
polynomial time.  Thus, we partition $V(D_1)$ into $D_i$'s and
$N_i$'s, where each $D_1[N_i]$ has independence number at most
$\alpha-1$.  We color each $D_i$ in polynomial time and recurse on the
$N_i$s. Since no vertex is considered in more than one recursive call
for a given $\alpha$, this leads to at most $\alpha n \leq n^2$
recursive calls, resulting in a polynomial-time termination of
Algorithm \ref{alg:2AlphaColoringLight}.
\end{proof}

\subsubsection{Coloring $C_3$-Free Dense Digraphs}

In this section, $D = (V,A)$ is a $C_3$-free oriented graph with
independence number $\alpha(D)$.  Recall that $C_k$ denotes a cyclicly
oriented $k$-cycle.  Since adding arcs does not decrease the
dichromatic number or increase the independence number, we can assume
that $D$ is \emph{maximally triangle-free}, which means that for any
pair of vertices $u,v$ in $V$ with no arc between them, both arcs $uv$
and $vu$ would each create a $C_3$ in $D$.  In other words, if there
is no arc between $u$ and $v$, then there is a $C_4$ in $D$ containing
$u$ and $v$.  We will call such a digraph a {\em maximally $C_3$-free
  digraph}.  We denote the total out-neighborhood (respectively,
in-neighborhood) of a vertex subset $S$ by $\NFP(S)$ (respectively,
$\NFM(S)$).  A vertex $w \notin S$ belongs to $\NFP(S)$ if for all $s
\in S$, we have arc $sw \in A$.  The goal of this section is to prove
the following theorem.

\begin{theorem}\label{thm:c3-color}
Let $D$ be a $C_3$-free digraph with independence number $\alpha$.
Then $D$ can be colored with at most $\frac{(\alpha+8)!}{9!}$ colors
in polynomial time.
\end{theorem}

This improves on the previous bound of $\chiv(D) \leq
35^{\alpha-1}\alpha!$ proved by \cite{harutyunyan2019coloring}.
Notice that if $D$ is a $C_3$-free digraph and $\alpha(D) = 1$, then
$D$ is acyclic (i.e., $\chiv(D) = 1$).
To prove this result, we will again use the \ndecomp presented in
Definition \ref{def:vc_gen_dig}.  When $D$ is a $C_3$-free digraph,
this decomposition has additional useful properties.

\begin{observation}\label{obs:savingsOnB}
Let $D$ be a $C_3$-free digraph.  For a vertex chain $(v_i)_{0 \leq
  i \leq k}$ and a corresponding path decomposition in $D$, there is
  no arc from $D_0$ to $D_{k+1}$.
\end{observation}

\begin{proof}
Consider a vertex $v \in D_{k+1}$.  Since $v \notin N_0$ and $v
\notin D_0$, it follows that $v \in N^-(v_0)$.  Thus, for some $u \in
D_0$, if there is an arc $uv$, then there is a $C_3$ in $D$, a
contradiction.
  \end{proof}

By Lemma \ref{lem:eff_loc_to_glob_dig},
Observation \ref{obs:savingsOnB} allows us to color $D_0$ and
$D_{k+1}$ with the same color palette.  This is useful when $D_0$ and
$D_{k+1}$ require many more colors than $N(e_i)$.  Let $g(\alpha)$
denote the maximum dichromatic number of a $C_3$-free digraph with
independence number at most $\alpha$.

\begin{lemma}\label{lem:non-neighbor}
Let $D$ be a maximally $C_3$-free digraph.  If $uv$ is a
non-arc in $D$, then $\chiv(N(uv)) \leq g(\alpha-1)$.
\end{lemma}

\begin{proof}
Let $a$ be a vertex in $N(uv)$.  Since $uv$ is a non-arc, there is
also no edge between $u$ and $v$.  Thus, it must be the case that
adding arc $vu$ causes a $C_3$, say with $w$.  So then
$N(uv) \subseteq N^o(w)$.  If this were not the case, then we would
have either triangle $wva$ or $wau$, which is a contradiction.
\end{proof}

\begin{lemma}\label{lem:good-pair}
  Let $D$ be a maximally $C_3$-free digraph with independence number
  at most $\alpha$.  $D$ has vertices $s_1,s_2 \in V$ such that
  $\chiv(N^-(s_1)) \leq \alpha \cdot g(\alpha-1)$ and
  $\chiv(N^+(s_2)) \leq \alpha \cdot g(\alpha-1)$.
\end{lemma}


\begin{proof}
Find a semi-kernel $K$ in $D$.  A semi-kernel is an independant and
distance-2 dominating set in a digraph and can be found in polynomial
time~\cite{bondy2003short}.  First, we claim that $\NFM(K)$ is
empty. If not, then there is a vertex $y\in \NFM(K)$, and $y$ is not
dominated by any vertex in $K$ (since our digraph does not contain
digons). Thus there exists a vertex $v\in K$ such that $v$ 2-dominates
$y$, while $yv$ is an arc by definition of $\NFM(K)$ which gives a
$C_3$ in $D$, a contradiction.
	
Now denote by $U$ the set of vertices not in any non-neighborhood
$N^o(v)$ of a vertex $v \in K$.  In other words, a vertex $u$ is in
$U$ if $u$ has an arc to or from each vertex in $K$.  For each vertex
$u\in U$, define $N_u=N^-(u)\cap K$, which gives a subset family of
size $|U|$ on $K$. Among them pick a minimal one, say $N_{u'}$, and
then arbitrarily choose a vertex $s_1\in N_{u'}$. We claim that
$\chiv(N^-(s_1)) \leq \alpha \cdot g(\alpha-1)$.

For any in-neighbor $w$ of $s_1$, either $w$ lies in non-neighborhood
of some vertex in $K-s_1$, or $w\in U$. The former case can be covered
by $(\alpha -1)g(\alpha-1)$ colors since $|K|\le
\alpha$. Next we prove that in the latter case, $w\in N^o(u')$.  Note that
$ws_1$ and $s_1u'$ are arcs by definition.  As in the proof of Lemma
\ref{lem:non-neighbor}, it suffices to find $v'$ such that $v'w$ and
$u'v'$ are arcs.  We claim that there must exist such a $v' \in K$.
Suppose not, then $N_w\subsetneq N_{u'}$, contradicting the fact that
$N_{u'}$ is minimal since $w,u'\in U$ and $s_1 \in N_{u'}$ but $s_1
\notin N_w$.  Thus $\chiv(N^-(s_1))\le (\alpha
-1)g(\alpha-1)+g(\alpha-1) = \alpha \cdot g(\alpha-1)$.
	
By reversing the direction of all arcs, the same argument results in a
vertex $s_2$ such that $\chiv(N^+(s_2)) \leq \alpha \cdot
g(\alpha-1)$. (Note that $s_1$ and $s_2$ need not be distinct.)
\end{proof}

The following corollary of \Cref{lem:good-pair} that shall be useful
in simplifying the presentation of Algorithm~\ref{alg:T-freeColoring}.

\begin{corollary}\label{C:good-pair}
Let $D$ be a maximally $C_3$-free digraph with independence number at
most $\alpha$. Then, in polynomial time, we can find vertices $s,t$
such that each $N^+(s)$ and $N^-(t)$ can be partitioned into at most
$\alpha$ components $S_1,\ldots,S_\alpha$ and $T_1,\ldots,T_\alpha$,
respectively, such that for $1\leq j\leq \alpha$, $D[S_j]$ and
$D[T_j]$ each has independence number at most $\alpha-1$. 
\end{corollary}

We are now ready to prove Theorem \ref{thm:c3-color} via the algorithm
whose steps are detailed in Algorithm~\ref{alg:T-freeColoring}.

\begin{proof}[Proof of \Cref{thm:c3-color}]
We give an overview of Algorithm~\ref{alg:T-freeColoring}.  Consider a
maximally $C_3$-free digraph $D$ with independence number $\alpha$.
By Lemma \ref{lem:good-pair}, we can find a pair of vertices $s,t$
such that $\chiv(N^+(s)) \leq \alpha \cdot g(\alpha-1)$ and
$\chiv(N^-(t)) \leq \alpha \cdot g(\alpha-1)$.  Next, we find a
$b$-vertex chain $(v_i)_{0 \leq i \leq k}$, where $v_0 = s, v_k = t$
and $b = \alpha \cdot g(\alpha-1)$.  For every non-arc $e_i$ in the
corresponding edge-chain, we have $\chiv(N(e_i)) \leq g(\alpha-1)$ due
to Lemma \ref{lem:non-neighbor}.  For every arc $e_i$, observe that
$N(e_i)$ is empty, since $D$ is $C_3$-free.  Thus, for each $D_i$ for
$1\leq i \leq k$ in the path decomposition, we have $\chiv(D_i) \leq
g(\alpha-1)$.  Additionally, each $D[N_i]$ has independence number at
most $\alpha-1$ by Observation \ref{obs:alphaMinusOne} and thus has
$\chiv(N_i) \leq g(\alpha-1)$.  Applying
Lemma \ref{lem:eff_loc_to_glob_dig} and
Observation \ref{obs:savingsOnB}, implies the following claim.

\begin{claim}\label{clm:c3-color}
$g(\alpha) \leq (\alpha+8) g(\alpha-1)$.
\end{claim}
The base case is $g(1) = 1$, since a $C_3$-free tournament is acyclic.
A simple calculation leads to the upper bound stated in the theorem.
It remains to show that Algorithm~\ref{alg:T-freeColoring} runs in
polynomial time.  Since an execution of the algorithm on $D$ makes at
most $\alpha + 8$ recursive calls on subdigraphs with independence
number $\alpha-1$, and these respective subdigraphs are vertex
disjoint, each vertex is included in at most $\alpha n \leq n^2$
recursive calls, resulting in a polynomial-time termination.
\end{proof}

\begin{algorithm}
\caption{Dicoloring$C_3$-Free($D,\alpha,P$)}
\label{alg:T-freeColoring}
Input: {A maximally $C_3$-free digraph $D$ with independence number at most $\alpha$ and a palette $P$ of $\frac{(\alpha+8)!}{9!}$
colors.}

Output: A coloring of $V(D)$ using palette $P$. 
\begin{enumerate}
\item If $D$ is a tournament, then $D$ is acyclic, so color using $1$ color from $P$ and return. 
\item Partition $P$ arbitrarily into $\alpha+8$ palettes $P_0,\ldots,P_{\alpha+7}$ such that for $0\leq i \leq \alpha+7$, $|P_i| = g(\alpha-1) = \frac{(\alpha+7)!}{9!}$.
\item Find two vertices $s,t$ such that $\chiv(N^+(s)) \leq \alpha \cdot g(\alpha-1)$ and $\chiv(N^-(t))\leq \alpha \cdot g(\alpha-1)$. (\Cref{lem:good-pair})

\item Find a $\alpha \cdot g(\alpha-1)$-vertex chain between from $s$ to $t$. Let the chain be $(v_i)_{0\leq i \leq k}$ where $v_0=s$ and $v_k=t$.

\item For $0\leq i\leq k+1$:
\begin{enumerate}

\item If $N_i \neq \emptyset$: Dicoloring$C_3$-Free($D[N_i],\alpha-1, P_{i \mod 4}$).
\item If $1\leq i \leq k$ and $D_i \neq \emptyset$: Dicoloring$C_3$-Free($D[D_i],\alpha-1, P_{4+(i \mod 5)}$).

\end{enumerate}
\item Partition $D_0 \subseteq N^+(s)$ into $\alpha$ components
$S_1,\ldots,S_\alpha$ and $D_{k+1} \subseteq N^-(t)$ into
$T_1,\ldots,T_\alpha$ such that the independence number of each
$D[S_j]$ and $D[T_j]$ is at most $\alpha-1$. (Corollary \ref{C:good-pair}.)
\item If $S_1 \neq \emptyset$:  Dicoloring$C_3$-Free($D[S_1],\alpha-1, P_{4}$).
\item If $T_1\neq \emptyset$:  Dicoloring$C_3$-Free($D[T_1],\alpha-1, P_{4+(k+1 \mod 5)}$).
\item For $2\leq j\leq \alpha$:
\begin{enumerate}
    \item If $S_j \neq \emptyset$:  Dicoloring$C_3$-Free($D[S_j],\alpha-1, P_{j+7}$).
    \item If $T_j \neq \emptyset$:  Dicoloring$C_3$-Free($D[T_j],\alpha-1, P_{j+7}$).
\end{enumerate}
\end{enumerate}

\end{algorithm}

We conclude this section with the following conjecture.

\begin{conjecture}\label{AlphaConj}
In an oriented 
$C_3$-free digraph $D$ with independence number
$\alpha$, there is a vertex $v$ such that $D[N^+(v)]$ has
independence number $\alpha-1$.  \end{conjecture}

A proof of this conjecture would imply a singly exponential bound
on the dichromatic number for $C_3$-free digraphs (i.e., of the form $c^{\alpha}$).  Specifically, in Claim \ref{clm:c3-color}, we would have $g(\alpha) \leq 9 \cdot g(\alpha-1)$.  Conjecture \ref{AlphaConj} is trivially true for $\alpha=1$, since a $C_3$-free tournament contains a sink vertex, which has out-degree zero.  It is also true when $\alpha=2$ due to a proof by Pierre Charbit and Samuel Coulomb.

\begin{theorem}\label{AlphaTwo}
In an oriented $C_3$-free digraph $D$ with independence number $2$, there is a vertex $v$ such that $D[N^+(v)]$ is a tournament.
\end{theorem}

\begin{proof}
Let $x$ be a vertex in $D$ with minimum out-degree and let $y$ be the
sink of $N^o(x)$ (the non-neighbors of $x$, which by
Observation \ref{obs:alphaMinusOne} form a tournament).  We will show
that either $N^+(x)$ or $N^+(y)$ is a tournament.

Suppose $y$ has an out-neighbor $z \in N^+(x)$. Notice that $z$
cannot have an out-neighbor $w \in N^o(x)$, because $yzw $ would form a
$C_3$.  Moreover, $z$ cannot have an out-neighbor $u \in N^-(x)$, since $xzu$
would also form a $C_3$.  Thus, all out-neighbors of $z$ are in
$N^+(x)$, which implies that the out-degree of $z$ is strictly less
than the out-degree of $x$, a contradiction. Hence, $y$ has no out-neighbor in $N^+(x)$.

Now, suppose $y$ has an in-neighbor $z \in N^+(x)$.  Recall that $y$ has no out-neighbor in $N^+(x)$ nor in $N^o(x)$, so $N^+(y) \subseteq N^-(x)$. But for any vertex $w \in N^+(y)$, there is no arc $wz$ and no arc $zw$, as otherwise $xzw$ or $ywz$ form a $C_3$. Thus, $N^+(y)$ belongs to $N^o(z)$ and must therefore be a tournament.
Lastly, if $y$ has no neighbor in $N^+(x)$, then $N^+(x) \subseteq N^o(y)$ is a tournament.
\end{proof}

In the special case of a
$C_3$-free digraph $D$ where $\alpha(D) = 2$,
\cite{harutyunyan2019coloring} claimed a bound of around 25 colors,
although the best lower bound is 2 (for $C_4$ and $C_5$).  
As a corollary of Theorem \ref{thm:c3-color}, we can improve this bound to 10 colors.  Using Theorem \ref{AlphaTwo}, we can further decrease this bound by one.

\begin{corollary}
Let $D$ be a $C_3$-free digraph with $\alpha(D) = 2$.  Then $\chiv(D)
\leq 9$.
\end{corollary}

\section{Conclusions and Open Questions}

We established that we can efficiently color a $2$-colorable digraph
$D$ on $n$ vertices using at most $2\sqrt{n}$ colors.  In the case of
undirected graphs, we can color $3$-colorable graphs using
$O(n^{0.19747})$ colors~\cite{kawarabayashi2024better}.  An
interesting open question is if we can find an efficient algorithm to
color a $2$-colorable digraph with $O(n^{\frac{1}{2}-\eps})$ colors
for some $\eps>0$.  As mentioned in Section~\ref{sec:algorithms}, a
positive answer to \Cref{KMSQuestion} will lead to a better coloring
for $2$-colorable digraphs.  Another interesting question in this
direction can be to see if semidefinite programming based techniques
be used here in combination with the combinatorial arguments to obtain
better colorings.

Regarding inapproximability results on 2-colorable digraphs, which are
rather limited, the best result we know is that it is NP-hard to color
a $2$-colorable digraph (actually, a tournament) using at most $3$
colors~\cite{klingelhoefer2024coloring}. One natural question is the
following.
\begin{question}
    For a constant $k>3$, what is the computational complexity of coloring a $2$-colorable digraph using at most $k$ colors?
\end{question}
We suspect that the above question has an answer similar to the same question on $3$-uniform hypergraphs: Given a $3$-uniform hypergraph that is $2$-colorable, it is \NP-hard to color it using $k$ colors, for any constant $k$~\cite{DRS05}.


\section{Acknowledgments} We thank Pierre Charbit and Samuel Coulomb for communicating their nice proof of Theorem \ref{AlphaTwo} and allowing us to include it in our paper.  We also acknowledge the Dagstuhl Workshop ``Graph Algorithms: Cuts, Flows, and Network Design (Seminar 23422)'' where some work was initiated.

\bibliographystyle{alpha}
\bibliography{ref}
\end{document}